\RequirePackage{fix-cm}
\documentclass[twocolumn]{svjour3}          
\smartqed  

\usepackage{graphicx}
\usepackage{mathptmx}      
\usepackage{lipsum}
\usepackage{amsmath}
\usepackage{amsfonts}
\usepackage{algorithmic}
\usepackage[ruled]{algorithm2e}
\journalname{Journal of Scheduling}
\begin{document}

\title{Maximizing Online Utilization with Commitment\thanks{A preliminary version of the preemptive part of this paper was published at the European Symposium of Algorithms ESA 2016~\cite{SchwiegelshohnS16}}
}



\author{Chris Schwiegelshohn         \and
        Uwe Schwiegelshohn
}

\authorrunning{Schwiegelshohn, C. and Schwiegelshohn, U.} 

\institute{Chris Schwiegelshohn \at
              Sapienza, University of Rome, Italy \\
              \email{schwiegelshohn@diag.uniroma1.it}           \\
           \and
           Uwe Schwiegelshohn \at
              TU Dortmund University, Germany \\
              \email{uwe.schwiegelshohn@tu-dortmund.de}
}

\date{Received: date / Accepted: date}

\maketitle

\begin{abstract}
We investigate online scheduling with commitment for parallel identical machines. Our objective is to maximize the total processing time of accepted jobs. As soon as a job has been submitted, the commitment constraint forces us to decide immediately whether we accept or reject the job. Upon acceptance of a job, we must complete it before its deadline $d$ that satisfies $d \geq (1+\varepsilon)\cdot p + r$, with $p$ and $r$ being the processing time and the submission time of the job, respectively while $\varepsilon>0$ is the slack of the system. Since the hard case typically arises for near-tight deadlines, that is $\varepsilon\rightarrow 0$, we consider $\varepsilon\leq 1$. We use competitive analysis to evaluate our algorithms. While there are simple online algorithms with optimal competitive ratios for the single machine model, little is known for parallel identical machines. 

Our first main contribution is a deterministic preemptive online algorithm with an almost tight competitive ratio on any number of machines. For a single machine, the competitive factor matches the optimal bound $\frac{1+\varepsilon}{\varepsilon}$ of the greedy acceptance policy. Then the competitive ratio improves with an increasing number of machines and approaches $(1+\varepsilon)\cdot\ln \frac{1+\varepsilon}{\varepsilon}$ as the number of machines converges to infinity. This is an exponential improvement over the greedy acceptance policy for small $\varepsilon$. 

In the non-preemptive case, we present a deterministic algorithm on $m$ machines with a competitive ratio of $1+m\cdot \left(\frac{1+\varepsilon}{\varepsilon}\right)^{\frac{1}{m}}$.
This matches the optimal bound of  $2+\frac{1}{\varepsilon}$ of the greedy acceptance policy for a single machine while it again guarantees an exponential improvement over the greedy acceptance policy for small $\varepsilon$ and large $m$. In addition, we determine an almost tight lower bound that approaches $m\cdot \left(\frac{1}{\varepsilon}\right)^{\frac{1}{m}}$ for large $m$ and small $\varepsilon$.
\end{abstract}

\section{Introduction}
\label{intro}
We consider an online resource allocation and scheduling problem modeled as follows. Given a set of $m$ identical machines, we process a sequence of jobs in an online fashion. Each job has a release date $r_j$, a deadline $d_j$, and a processing time $p_j$. Our objective is to maximize the total utilization of the system, that is the sum of the processing times of all completed jobs, subject to a \emph{commitment} guarantee. Due to this guarantee we must decide immediately upon submission of a job, whether to accept or reject it. In case of acceptance, we must complete it before its deadline and cannot accept any other job that would prevent us from doing so.

This problem has received a significant amount of attention and tight algorithms with competitive ratios depending on $\Delta=\log \frac{\max p_j}{\min p_j}$ are known for a single machine without further assumptions on the input~\cite{LiT94,GPS00}. To obtain more meaningful results, a number of papers use a slack parameter $\varepsilon>0$ and require $d_j\geq (1+\varepsilon)\cdot p_j + r_j$ for all jobs. 
Typically, the hard case arises for small slack ($\varepsilon\rightarrow 0$).

We consider the preemptive and the non-preemptive versions of the problem for small slack values $\varepsilon\leq 1$. In case of preemption, we allow migration and assume no penalty for preemption and migration. 

We use the common approach of competitive analysis to measure the performance of online algorithms, that is, we determine the competitive ratio. Since we are facing a maximization problem, the strict competitive ratio is the maximum value of $\frac{\text{OPT}({\mathcal{J}})}{\text{ALG}({\mathcal{J}})}$ over all input sequences ${\mathcal{J}}$, with $\text{OPT}({\mathcal{J}})$ and $\text{ALG}({\mathcal{J}})$ being the objective values of an optimal offline algorithm and the examined online algorithm, respectively.

\subsection{Our Contribution}

We show upper and lower bounds on the competitive ratio of utilization maximization with commitment in dependence of the slackness parameter $\varepsilon$. Specifically:
\begin{itemize}
\item For $m$ machines supporting preemption and migration, we give a deterministic algorithm with a strict competitive ratio of
\[m\cdot (1+\varepsilon)\cdot \left(\left(\frac{1+\varepsilon}{\varepsilon} \right)^{\frac{1}{m}}-1 \right).\]
This bound decreases with increasing $m$, approaching $(1+\varepsilon)\cdot \ln(\frac{1+\varepsilon}{\varepsilon})$ as $m$ tends to infinity and matching the competitive ratio $\frac{1+\varepsilon}{\varepsilon}$ of the optimal greedy algorithm for $m=1$.
\item We provide an almost matching lower bound of 
\[\lfloor m\cdot (1+\varepsilon)\rfloor\cdot \left(\left(\frac{1+\varepsilon}{\varepsilon} \right)^{\frac{1}{m}}-1 \right)
.\]
\item For $m$ machines without preemption, we give a deterministic algorithm with a strict competitive ratio of 
\[m \cdot\left(\frac{1+\varepsilon}{\varepsilon}\right)^{\frac{1}{m}}+1. \]
For $ m> \ln{\frac{1+\varepsilon}{\varepsilon}}$, we apply partitioning to improve this competitive ratio. In particular, we obtain a competitive ratio $e\ln \frac{1+\varepsilon}{\varepsilon} + 1$  for integral values $ \ln \frac{1+\varepsilon}{\varepsilon}$ and $m=i\cdot \ln \frac{1+\varepsilon}{\varepsilon}$, $i\in \mathbb{N}$. For small slack $\varepsilon$, these results significantly improve over the state of the art greedy algorithm with a competitive ratio of $2+\frac{1}{\varepsilon}$.
\item For the non-preemptive machine environment, we also give the first lower bounds  that are asymptotically greater than $1$ for any number of machines. In particular, for large $m$ and $\varepsilon\leq 1$, our lower bound approaches 
\[m\cdot \left(\frac{1}{\varepsilon} \right)^{\frac{1}{m}}\] 
which is almost tight (up to a $(1+\varepsilon)^{1/m}$ factor and the additive value 1).
\end{itemize}

\subsection{Related Work}
Using the three field notation of Graham et al.~\cite{Graham}, the $P|pmtn,r_j|\sum (1-U_j)w_j$ problem admits a pseudopolynomial time algorithm, see Lawler~\cite{Lawler90}, and a FPTAS due to a technique by Pruhs and Woeginger~\cite{PruhsW07}. We are not aware of any simplifications when assuming $w_j=p_j$. For preemption without migration, these ideas were further extended to admit a $2+\varepsilon$ approximation~\cite{AwerbuchAFLR01}. A reduction from preemptive to non-preemptive due to Kalyanasundaram and Pruhs~\cite{KaP01} yields a $6+\varepsilon$ approximation. Recently, Alon et al.~\cite{AAB18} analyzed the gap between schedules with a bounded number of $k$ preemptions and an unbounded number of preemptions. They show that the utilization ratio is at most $O(\log_{k+1} n, \log_{k+1} \frac{\max p_j}{\min p_j})$, which is asymptotically optimal.

Variants of the online problem have been studied in real-time scheduling. Without commitment and assuming preemptions, there exists a tight $4$-competitive deterministic algorithm on a single machine for a class of well-behaved payoff functions including $w_j = p_j$, see, for instance, Baruah et al.~\cite{BKMMRRSW92} and W\"{o}ginger~\cite{Woe94}.  Kalyanasundaram and Pruhs~\cite{KaP03} gave a $O(1)$ competitive randomized algorithm for maximizing the number of completed jobs and showed that no constant competitive deterministic algorithm exists.
For general pay-off functions in a non-preemptive environment, several authors~\cite{CaI98,KoS94,Por04} showed finite competitive ratios of $O(\min(\log \frac{\max p_j}{\min p_j},\log \frac{\max w_j}{\min w_j})$ if $d_j=r_j+p_j$ holds. For slack\footnote{In literature, the stretch $f = 1+\varepsilon$ is often used instead. The two notions are equivalent.} $\varepsilon>0$, Lucier et al.~\cite{LMNY13} proposed an algorithm with competitive ratio $2+\Theta\left(\frac{1}{\sqrt[3]{1+\varepsilon}-1}\right) +  \Theta\left(\frac{1}{(\sqrt[3]{1+\varepsilon}-1)^2}\right)$.

Commitments make the problem much more difficult. Indeed, Lucier et al.~\cite{LMNY13} showed that even with slack, no performance guarantees are possible for general weights $w_j$. By proposing a model with limited commitment called $\beta$-responsive, Azar et al.~\cite{AKLMNY15} gave an algorithm with competitive ratio $5+\Theta\left(\frac{1}{\sqrt[3]{1+\varepsilon}-1}\right) +  \Theta\left(\frac{1}{(\sqrt[3]{1+\varepsilon}-1)^2}\right)$ for $\varepsilon>3$. Similar relaxed commitment models were recently considered by Chen et al.~\cite{CEMSS18} for the problem of maximizing the number of accepted jobs $(w_j=1)$. For $w_j=p_j$ (the utilization maximization problem), a number of algorithms with commitment have been proposed. On a single machine without assuming slack or preemption, Lipton and Tomkins~\cite{LiT94} gave a lower bound of $\Omega (\log \Delta)$ for the competitive factor of any randomized online algorithm. Additionally assuming $d_j=p_j+r_j$ for all jobs, they obtained an upper bound of $O(\log^{1+\delta} \Delta)$ for any $\delta>0$. For arbitrary deadlines, Goldman, Parwatikar, and Suri~\cite{GPS00} presented a $6(\lceil \log \Delta\rceil +1)$ competitive algorithm. 

Baruah and Haritsa~\cite{BaH97} were among the first to present results for problems with slack when they addressed the single machine problem with preemption and commitment. For this problem, they gave an algorithm with a competitive ratio $1+\frac{1}{\varepsilon}$ and a lower bound of $1+\frac{1}{\lceil 1+ \varepsilon \rceil}$. For the corresponding non-preemptive problem on a single machine, Goldwasser's algorithm~\cite{Gol99,Gol03} guarantees a tight deterministic competitive ratio of $2+\frac{1}{\varepsilon}$.  
A similar competitive ratio is also given by Garay et al.~\cite{GNYZ02}.

For parallel identical machines, DasGupta and Palis~\cite{DaP01} consider preemption without migration and suggest a simple greedy algorithm. They claim the same competitive factor as for the single machine problem with preemption. They also give a lower bound of $1+\frac{1}{m\varepsilon}$. For the non-preemptive case, Kim and Chwa~\cite{KiC01} apply Goldwasser's algorithm~\cite{Gol99,Gol03} to parallel identical machines. Lee~\cite{Lee03} considers a model without commitment and presents a non-preemptive algorithm that adds many newly submitted jobs to a queue and decides later whether to schedule or to reject them. He claims a competitive ratio of $m+1+m\varepsilon^{-1/m}$ for small slack values. This competitive ratio is slightly worse than our competitive ratio for the same problem with commitment. 

\section{Notations and Basic Properties}
\label{sec:notations}

We schedule jobs of a job sequence $\mathcal{J}$ on $m$ \textit{parallel identical machines} ($P_m$ model). \textit{Release date} $r_j$,  \textit{processing time} $p_j$, and \textit{deadline} $d_j$ of a job $J_j\in {\mathcal{J}}$ obey the \textit{slack condition} $d_j-r_j\geq (1+\varepsilon)\cdot p_j$. $\varepsilon>0$ is the fixed \textit{slack} factor of the system. We say that $J_j$ has a \textit{tight slack} if $d_j-r_j=(1+\varepsilon)\cdot p_j$ holds. In the three field notation, the constraint $\varepsilon$ indicates a problem with slack factor $\varepsilon$. We are interested in problems with a small slack ($\varepsilon\leq 1$). The job sequence $\mathcal{J}$ defines an order of the jobs such that job $J_i$ precedes job $J_j$ for $r_i<r_j$.   

We discuss an \textit{online scenario with commitment}: the jobs arrive in sequence order and we must decide whether to accept the most recent job or not before seeing any job at a later position in the sequence. Complying with the common notation in scheduling theory, we use the binary variable $U_j$ to express this decision for job $J_j$: $U_j=0$ denotes acceptance of job $J_j$ while the rejection of $J_j$ produces $U_j=1$. We commit ourselves to complete every accepted job not later than its deadline and say that a schedule is \textit{valid} if it satisfies this commitment. Therefore, we can only accept a new job if there is a \textit{valid} extension of the current schedule that includes this job. We consider \textit{preemptive} and \textit{non-preemptive} variants of this problem. If the system allows preemption (\textit{pmtn}) then we can interrupt the execution of any job and immediately or later resume it on a possibly different machine without any penalty. In the non-preemptive case, we combine the acceptance commitment with a machine and start time commitment (immediate decision, see Bunde and Goldwasser~\cite{BuG10}). 

For a previously released job, we only consider its remaining processing time when generating or modifying a preemptive schedule. Therefore, progression of time may change the values of variables that depend on the processing time. For such a variable, we use the appendix $|_t$ to indicate its value at time $t$.

It is our goal to maximize the total processing time of all accepted jobs ($\sum p_j\cdot (1-U_j)$). An online algorithm $ALG$ determines $U_j(ALG)$ for all jobs in $\mathcal{J}$ such that there is a valid schedule $S$ for all jobs $J_j$ with $U_j(ALG)=0$.
$P[t,t')(S,{\mathcal{J}})$ denotes the total processing time of schedule $S$ for job sequence ${\mathcal{J}}$ in interval $[t,t')$ while $P^*[t,t')({\mathcal{J}})$ is the maximum total processing time of any optimal schedule for the same job sequence in the same interval. We evaluate our algorithm by determining bounds for the competitive ratio, that is the largest ratio between the optimal objective value and the objective value produced by online algorithm $ALG$ for all possible job sequences $\mathcal{J}$: $ALG$ has a strict competitive ratio $c_{ALG}$ if $c_{ALG}\geq \frac{P^*[0,\max_{J_j\in\mathcal{J}}d_j)({\mathcal{J}})}{P[0,\max_{J_j\in\mathcal{J}}d_j)(S,{\mathcal{J}})}$ holds for any job sequence $\mathcal{J}$. Finally, $c$ is a lower bound of the competitive ratio if $c\leq c_{ALG}$ holds for any online algorithm $ALG$.

\section{Online Scheduling with Preemption}
\label{sec:pmtn}

We present an approach for solving our problem on parallel identical machines with preemption. Contrary to previous approaches, our algorithm is based on lazy acceptance, that is, we may not accept a job although its acceptance allows a valid schedule. We show that our lazy acceptance leads to better competitive ratios than the greedy acceptance approach.

In this section, we frequently use function $\left( \frac{1+\varepsilon}{\varepsilon}\right)^x$ with the identity
\begin{align} 
\label{eq:func_add}
\sum_{j=i+1}^{m+i}\left( \frac{1+\varepsilon}{\varepsilon}\right)^{\frac{j}{m}} &= \sum_{j=i}^{m+i-1}\left( \frac{1+\varepsilon}{\varepsilon}\right)^{\frac{j}{m}}+\frac{1}{\varepsilon} \cdot \left( \frac{1+\varepsilon}{\varepsilon}\right)^{\frac{i}{m}}
\end{align} 
and the threshold expression
\begin{align}
\label{eq:min_parameter}
f(m,\varepsilon) & := \frac{\varepsilon}{1+\varepsilon}\cdot \sum_{j=0}^{m-1}\left( \frac{1+\varepsilon}{\varepsilon}\right)^{\frac{j}{m}} = \frac{1}{1+\varepsilon}\cdot \frac{1}{\left(\frac{1+\varepsilon}{\varepsilon}\right)^{\frac{1}{m}}-1}.
\end{align}

We start by recalling a known aggregation function for a sequence $\mathcal{J}$ of jobs that are submitted at time $t$ or earlier but not yet completed ($p_j|_t>0$): $V_{min}(\tau)|_t$ for $\tau\geq t$ is the minimum total processing time that we must execute in interval $[t,\tau)$ of any valid schedule.
\begin{align*}
V_{min}(\tau)|_{t} & := \sum_{J_j\in \mathcal{J}} \left\{ 
\begin{array}{lcl}
0 &  \mbox{ for}  & d_j-p_j|_{t}\geq \tau \\
p_j|_{t} &  \mbox{ for}  & \tau\geq d_j\\
\tau-(d_j-p_j|_{t}) &  \mbox{ else} \\
\end{array}
\right.
\end{align*}
Based on function $V_{min}$, Horn~\cite{Hor74} has given a necessary and sufficient condition for the existence of a valid preemptive schedule:
\begin{theorem}[Horn] 
\label{thm:horn}
There is a valid preemptive schedule in interval $[t,\max_{J_j\in\mathcal{J}}d_j)$ for a sequence $\mathcal{J}$ of jobs with deadlines on $m$ parallel identical machines if and only if we have  
\begin{align}
\nonumber
d_j & \geq t+p_j|_{t} \mbox{ for every job } J_j \mbox{ with } p_j|_t\geq 0 \mbox{ and }\\
\label{eq:schedulability}
V_{min}(\tau)|_{t} & \leq (\tau-t) \cdot m \mbox{ for every } \tau\in (t,\max_{J_j\in\mathcal{J}}d_j]. 
\end{align}
\end{theorem}

\begin{proof}
See Horn~\cite{Hor74}, proof of Theorem~2.
\end{proof}

Next we introduce two algorithms: the first algorithm decides whether to accept a newly submitted job while the second algorithm determines the schedule of the accepted jobs. Afterwards, we prove that the combination of both algorithms always produces a valid solution and give the competitive ratio of our approach. 

Main component of Algorithm~\ref{alg:online} \textit{Acceptance for Preemption} is the non decreasing threshold deadline $d_{min}$. We simply accept a new job if and only if its deadline is at least as large as $d_{min}$, see Lines 5 to 6 and 8 to 9. Initially, we assume $r_j\geq 0$ for all jobs and set $d_{min}=0$ (Line 1). Whenever there is a submission of a new job then Line 3 prevents $d_{min}$ to be in the past. Line 4 introduces a compensation value $V_{\Delta}$ for job parts with a large deadline that are scheduled earlier due to the availability of machines, compare Algorithm~\ref{alg:schedule}, lines 14-18, and Lemma~\ref{lem:Vmin_change}. Due to this early scheduling, progression of time may result in $V_{min}(d_{min})|_{r_j}<(d_{min}-r_j)\cdot f(m,\varepsilon)$. $V_{\Delta}$ compensates for this reduction and prevents an increase of the threshold that is too small. To this end, Line 7 sets $d_{min}$ to the largest value $\tau$ that satisfies $V_{min}(\tau)|_{r_j}+V_{\Delta}=(\tau-r_j)\cdot f(m,\varepsilon)$.

\begin{algorithm}[ht]
\caption{Acceptance for Preemption}
\label{alg:online}
\begin{algorithmic}[1]
\STATE{$d_{min}=0$}
\FOR{each newly submitted job $J_j$}
\STATE{$d_{min}=\max\{d_{min},r_j\}$}
\STATE{$V_{\Delta}=(d_{min}-r_j)\cdot f(m,\varepsilon)-V_{min}(d_{min})|_{r_j}$}
\IF{$d_j\geq d_{min}$}
\STATE{accept $J_j$}
\STATE{$d_{min}=\operatorname{arg\,max}\{(\tau-r_j)\cdot f(m,\varepsilon)=V_{min}(\tau)|_{r_j}+V_{\Delta}\}$}
\ELSE
\STATE{reject $J_j$}
\ENDIF
\ENDFOR
\end{algorithmic}
\end{algorithm}

Algorithm~\ref{alg:schedule} \textit{Schedule Generation} uses the current time $t$ as input parameter and generates a schedule for interval $[t,T)$ including the specification of the end of the interval $T$ (Lines 6, 11, and 17). We must run the algorithm again if we have accepted a new job or as soon as time has progressed to $T$. We determine the largest deadline $d_k$ such that at most $m$ jobs contribute to $V_{min}(d_k)|_t$ (Line 2). If there is no such deadline, we skip any preallocation and select the setting to produce an LRPT schedule for the contributions of the jobs to $V_{min}(d_{1})|_t$ (Lines 3 to 7). Otherwise, we execute a preallocation by starting each of the jobs contributing to $V_{min}(d_ {k})|_t$ on a separate machine (Lines 9 and 10). $T$ exceeds $t$ by at most the smallest contribution among those jobs to $V_{min}(d_{k})|_t$ (Line 11). $m_{LRPT}$ is the number of idle machines after this preallocation (Line 12). If there are idle machines in interval $[t,T)$ and the total number of uncompleted jobs exceeds $m$ then we apply LRPT to determine a schedule starting at time $t$ on these $m_{LRPT}$ idle machines for the contribution of any not yet allocated job to $V_{min}(d_{k+1})|_t$ (Lines 15 and 16). If a machine becomes idle in this LRPT schedule before the previously determined value of $T$ then we reduce $T$ accordingly (Line 17). 
\begin{algorithm}[t]
\caption{Schedule Generation$(t)$}
\label{alg:schedule}
\begin{algorithmic}[1]
\STATE{index the deadlines such that $d_0=t<d_1< \ldots <d_{\nu}=\max_{J_j\in\mathcal{J}}\{d_j\}$}
\STATE{$d_k$ is the largest deadline with at most $m$ jobs contributing to $V_{min}(d_k)|_t$}
\IF{there is no such deadline $d_k$}
\STATE{$k=0$}
\STATE{${\mathcal{J}}=\emptyset$}
\STATE{$T=d_1$}
\STATE{$m_{LRPT}=m$}
\ELSE
\STATE{${\mathcal{J}}=\{\mbox{all jobs contributing to }V_{min}(d_{k})|_t\}$}
\STATE{allocate each job $\in{\mathcal{J}}$ to a separate machine}
\STATE{$T=t+\min\limits_{J_j\in\mathcal{J}}\{\mbox{contribution of }J_j\mbox{ to }V_{min}(d_{k})|_t\}$}
\STATE{$m_{LRPT}=m-|{\mathcal{J}}|$}
\ENDIF
\IF{$k<\nu$ and $m_{LRPT}>0$} 
\STATE{${\mathcal{J}}_{LRPT}=\{\mbox{all jobs contributing to }V_{min}(d_{k+1})\}\backslash {\mathcal{J}}$}
\STATE{generate an LRPT schedule}
\STATE{$T=\min\{T,$ first idle time in the LRPT schedule$\}$}
\ENDIF
\RETURN{$T$}
\end{algorithmic}
\end{algorithm}
Since Algorithm~\ref{alg:schedule} prioritizes contributions to $V_{min}(\tau_1)|_t$ over additional contributions to $V_{min}(\tau_2)|_t$ for $\tau_2>\tau_1>t$, it generates a valid schedule if one exists, see Theorem~\ref{thm:horn}. 

Next we discuss the influence of progression of time due to Algorithm~\ref{alg:schedule} on $V_{min}$ and on $d_{min}$ for a valid schedule.
\begin{lemma}
\label{lem:Vmin_change}
If Algorithm~\ref{alg:schedule} has generated a valid schedule starting at time $t$ then progression from time $t$ to time $t'$ without acceptance of a new job in interval $(t,t')$ produces 
\begin{align}
\label{eq:Vmin_progression}
V_{min}(\tau)|_{t'} &\leq \frac{\tau-t'}{\tau-t} \cdot V_{min}(\tau)|_{t} \;\; \mbox{ for each } \tau\geq t' 
\end{align}
and for each $\tau > d_{min}|_{t'}$,
\begin{align}
\label{eq:dmin_progression}
V_{min}(\tau)|_{t'}-V_{min}(d_{min}|_{t'})& < (\tau-d_{min}|_{t'})\cdot f(m,\varepsilon).
\end{align}
\end{lemma}

\begin{proof}
Eq.~(\ref{eq:Vmin_progression}) holds for time $\tau>t'$ if no machine is idle in interval $[t,t')$ and every machine only executes job parts contributing to $V_{min}(\tau)|_{t}$  due to
\begin{align*}
m &\overset{Eq.(\ref{eq:schedulability})}{\geq} \frac{V_{min}(\tau)|_{t}}{\tau-t}\geq \frac{V_{min}(\tau)|_{t}-m\cdot (t'-t)}{\tau-t-(t'-t)}=\frac{V_{min}(\tau)|_{t'}}{\tau-t'}.
\end{align*}
If at least one machine executes in interval $[t,t')$ some job parts that do not contribute to $V_{min}(\tau)|_{t}$ for some time $\tau\geq t'$ or has some idle time then Algorithm~\ref{alg:schedule} guarantees that at time $t'$, there are $k<m$ jobs with uncompleted job parts contributing to $V_{min}(\tau)|_{t}$ and the total amount of processing time in interval $[t,t')$ that contributes to $V_{min}(\tau)|_{t}$ is at least $k\cdot (t'-t)$. Then we have
\begin{align*}
\frac{V_{min}(\tau)|_{t}}{\tau-t} & \geq \frac{V_{min}(\tau)|_{t'}+ k\cdot (t'-t)}{\tau-t'+t'-t} \geq \frac{V_{min}(\tau)|_{t'}}{\tau-t'}.
\end{align*}
We assume that Eq.~(\ref{eq:dmin_progression}) holds for $t$ and apply an inductive approach. For $d_{min}|_t\geq t'$, Algorithm~\ref{alg:online} sets $d_{min}|_{t'}=d_{min}|_t$. Due to  Algorithm~\ref{alg:schedule}, we have for each $\tau >d_{min}|_{t'}$,
\begin{align*}
V_{min}(\tau)|_{t'}-V_{min}(d_{min}|_{t'})& \leq V_{min}(\tau)|_{t}-V_{min}(d_{min}|_{t}) \\
& <(\tau-d_{min}|_{t})\cdot f(m,\varepsilon) \\
& < (\tau-d_{min}|_{t'})\cdot f(m,\varepsilon).
\end{align*}
For $d_{min}|_t< t'$, Algorithm~\ref{alg:online} sets $d_{min}|_{t'}=t'$ and $V_{\Delta}|_{t'}=0$. Then we have for each $\tau>d_{min}|_{t'}$,
\begin{align*}
V_{\Delta}|_{t'}+V_{min}(\tau)|_{t'} & = V_{min}(\tau)|_{t'} \overset{Eq.(\ref{eq:Vmin_progression})}{\leq} \frac{\tau-t'}{\tau-t}\cdot V_{min}(\tau)|_{t} \\
& \leq \frac{\tau-t'}{\tau-t}\cdot \left(V_{min}(\tau)|_{t}+V_{\Delta}|_t \right) \\ 
& <\frac{\tau-t'}{\tau-t}\cdot \left((\tau-t)\cdot f(m,\varepsilon) \right)\\
& < (\tau-d_{min}|_{t'})\cdot f(m,\varepsilon).
\end{align*}
\qed 
\end{proof}

Next, we prove that Algorithms~\ref{alg:online} and \ref{alg:schedule} guarantee the existence of a valid schedule. To this end, we introduce a new function $V(x)$ with $0\leq x$. In the middle part of this definition, we use $x=y\cdot\left(\frac{\varepsilon}{1+\varepsilon}\right)^{\frac{m-h}{m}}$ with $h\in\{0,1,\ldots,m-1\}$ and $1\leq y \leq\left(\frac{1+\varepsilon}{\varepsilon}\right)^{\frac{1}{m}}$. 
\begin{align}
\label{eq:Vx_definition}
V(x) := \left\{\begin{array}{ll}
x\cdot f(m,\varepsilon) &  \mbox{ for }   x \geq 1\\
\sum\limits_{i=0}^h \left(\frac{\varepsilon}{1+\varepsilon}\right)^{\frac{m-i}{m}} + x\cdot (m-h-1) & \mbox{ for }  1\leq x \leq \frac{\varepsilon}{1+\varepsilon}  \\
x\cdot m &  \mbox{ for }   \frac{\varepsilon}{1+\varepsilon}\geq x \geq 0 
\end{array} \right.
\end{align}

$V(x)$ is piecewise linear with non-negative slopes. Note that we have defined $V(x)$ twice at all corner points ($y=1$ and $x=1$). At each of these points, both definitions produce the same value. Therefore, $V(x)$ is continuous and we have
\begin{align}
\label{eq:Vx_value}
V(x_1) & \leq V(x_2)\;\; \mbox{ for } 0\leq x_1 < x_2 \\
\label{eq:Vx_constant}
V(x_1) & = V(x_2) \;\; \mbox{ only for } x_1,x_2\in \left[\left(\frac{\varepsilon}{1+\varepsilon} \right)^{\frac{1}{m}},1\right]\\
\label{eq:Vx_monotony}
m \geq \frac{V(x_1)}{x_1} & \geq \frac{V(x_2)}{x_2}\geq f(m,\varepsilon) \;\; \mbox{ for } 0<x_1<x_2 
\end{align}
Eq.~(\ref{eq:Vx_monotony}) holds since the slopes are non-negative and decreasing in interval $(0,1)$. Based on function $V(x)$, we claim the following property of our combination of algorithms:
\begin{lemma}
\label{lem:valid}
Algorithms~\ref{alg:online} and \ref{alg:schedule} guarantee  
\begin{align}
\label{eq:V_cond_bottom}
V_{min}(\tau)|_t & \leq (d_{min}-t)\cdot V\left(\frac{\tau-t}{d_{min}-t}\right)  & \mbox{ for } d_{min} >\tau \geq t 
\end{align} 
and for $\tau \geq d_{min}$, 
\begin{align}
\label{eq:V_cond_top} 
V_{min}(\tau)|_t & \leq V_{min}(d_{min})|_{t}+(\tau-d_{min})\cdot f(m,\varepsilon).  
\end{align}
\end{lemma}
 
$V_{min}(\tau)|_t  \leq (\tau-t)\cdot f(m,\varepsilon)-V_{\Delta}$ is equivalent to Eq.~(\ref{eq:V_cond_top}) for $\tau > d_{min}>t$.
\begin{proof}
We use an inductive approach by assuming validity of Eq.~(\ref{eq:V_cond_bottom}) and (\ref{eq:V_cond_top}) for time $t$ and separately consider the events \textit{progression of time} and \textit{acceptance of a new job}.

We start with progression from time $t$ to time $t'$ without acceptance of a new job in interval $(t,t')$. Eq.~(\ref{eq:V_cond_bottom}) and (\ref{eq:V_cond_top}) clearly hold for time $\tau=t'$. Eq.~(\ref{eq:dmin_progression}) guarantees validity of Eq.~(\ref{eq:V_cond_top}) for $\tau>d_{min}|_{t'}=\max\{d_{min}|_{t},t'\}$. For $t'<\tau\leq d_{min}|_{t}=d_{min}|_{t'}$, Eq.~(\ref{eq:V_cond_bottom}) also remains valid due to Lemma~\ref{lem:Vmin_change}:
\begin{align*}
V_{min}(\tau)|_{t'}&\overset{Eq.(\ref{eq:Vmin_progression})}{\leq} \frac{\tau-t'}{\tau-t}\cdot V_{min}(\tau)|_{t} \\
& \overset{Eq.(\ref{eq:V_cond_bottom})}{\leq} (\tau-t')\cdot\frac{d_{min}|_{t}-t}{\tau-t}\cdot  V\left(\frac{\tau-t}{d_{min}|_{t}-t}\right)\\
& \overset{Eq.(\ref{eq:Vx_monotony})}{\leq} (\tau-t')\cdot\frac{d_{min}|_{t}-t'}{\tau-t'}\cdot  V\left(\frac{\tau-t'}{d_{min}|_{t}-t'}\right) \\
& \leq (d_{min}|_t-t')\cdot  V\left(\frac{\tau-t'}{d_{min}|_t-t'}\right).
\end{align*}
 
The second event is the acceptance of a new job $J_j$ with $r_j=t$, $(1+\varepsilon)\cdot p_j \leq d_j-r_j$, and $d_j\geq d_{min}$, see Algorithm~\ref{alg:online}. To simplify the formal description of this event, we omit the appendix $|_t$ of the variables and apply a time shift by $-t$, that is, we assume $r_j=t=0$. This time shift does not affect the validity of our transformations since they always use $d_{min}-t$ or $\tau-t$.  

Again we assume the validity of Eq.~(\ref{eq:V_cond_bottom}) and (\ref{eq:V_cond_top}) before the submission of $J_j$. The acceptance of $J_j$ produces a new deadline threshold $d_{min}^{new}\geq d_{min}$. We partition time into several intervals: Eq.~(\ref{eq:V_cond_top}) is always valid for $\tau\geq d_{min}^{new}$ due to Line 7 of Algorithm~\ref{alg:online} and it remains valid for $d_{min}^{new}\cdot \left(\frac{\varepsilon}{1+\varepsilon}\right)^{\frac{1}{m}}\leq \tau < d_{min}^{new}$ due to Eq.~(\ref{eq:Vx_constant}) and $V_{min}(\tau)\leq V_{min}(d_{min}^{new})$. For $\tau\leq d_j-p_j$, Eq.~(\ref{eq:V_cond_bottom}) and (\ref{eq:V_cond_top}) remain valid  due to Eq.~(\ref{eq:Vx_monotony}) since there is no change in $V_{min}(\tau)$. 

For the remaining intervals, we start with the case $d_j=d_{min}=\frac{V_{min}(d_{min})+V_{\Delta}}{f(m,\varepsilon)}$. 

If job $J_j$ has tight slack ($p_j=\frac{d_{min}}{1+\varepsilon}$) then Line 7 of Algorithm~\ref{alg:online} produces  
\begin{align} 
\label{eq:dmin_tight}
d_{min}^{new} \geq \frac{V_{min}^{new}(d_{min})+V_{\Delta}}{f(m,\varepsilon)} & = \frac{V_{min}(d_{min})+V_{\Delta}+p_j}{f(m,\varepsilon)} \\ \nonumber
& = d_{min}+\frac{d_{min}}{(1+\varepsilon)\cdot f(m,\varepsilon)} \\
& \overset{Eq.(\ref{eq:min_parameter})}{=} 
d_{min}\cdot \left(\frac{1+\varepsilon}{\varepsilon} \right)^{\frac{1}{m}}
\end{align}
For $\tau=d_{min}\cdot y\cdot \left(\frac{\varepsilon}{1+\varepsilon} \right)^{\frac{m-h}{m}}$ with $h\in\{0,1,\ldots,m-1\}$ and $1\leq y \leq \left(\frac{1+\varepsilon}{\varepsilon} \right)^{\frac{1}{m}}$, there is 
\begin{align}
\nonumber
V_{min}^{new}(\tau) & = V_{min}(\tau)+\tau -d_{min}\cdot \frac{\varepsilon}{1+\varepsilon} \\ \nonumber
& \overset{Eq.(\ref{eq:V_cond_bottom})}{\leq} d_{min}\cdot V\left(\frac{\tau}{d_{min}} \right) + \tau - d_{min}\cdot \frac{\varepsilon}{1+\varepsilon} \\ \nonumber
 & \overset{Eq.(\ref{eq:Vx_definition})}{\leq} d_{min}\cdot \sum\limits_{i=0}^h \left(\frac{\varepsilon}{1+\varepsilon}\right)^{\frac{m-i}{m}} + \\ \nonumber 
& \hspace{40pt} + d_{min}\cdot y\cdot (m-h-1)\cdot \left(\frac{\varepsilon}{1+\varepsilon} \right)^{\frac{m-h}{m}} + \\ \nonumber
& \hspace{40pt} +d_{min}\cdot \left(y\cdot \left(\frac{\varepsilon}{1+\varepsilon} \right)^{\frac{m-h}{m}}-\frac{\varepsilon}{1+\varepsilon} \right)\\ \nonumber
& \leq d_{min} \cdot \sum\limits_{i=1}^h \left(\frac{\varepsilon}{1+\varepsilon}\right)^{\frac{m-i}{m}} + \\ 
\label{eq:Vnewmin_tight}
& \hspace{40pt} + d_{min} \cdot y \cdot (m-h)\cdot \left(\frac{\varepsilon}{1+\varepsilon} \right)^{\frac{m-h}{m}}.
\end{align}
For $h=0$, that is $d_{min}\cdot \frac{\varepsilon}{1+\varepsilon}\leq \tau \leq d_{\min}\cdot \left(\frac{1+\varepsilon}{\varepsilon} \right)^{\frac{1}{m}} \cdot \frac{\varepsilon}{1+\varepsilon}\leq d_{min}^{new}\cdot \frac{\varepsilon}{1+\varepsilon}$, we have
\begin{align*}
V_{min}^{new}(\tau) & \leq\tau \cdot m =\tau\cdot \frac{V\left(\frac{\tau}{d_{\min}\cdot \left(\frac{1+\varepsilon}{\varepsilon} \right)^{\frac{1}{m}}} \right)}{\frac{\tau}{d_{\min}\cdot \left(\frac{1+\varepsilon}{\varepsilon} \right)^{\frac{1}{m}}}} \\
& \leq d_{\min}\cdot \left(\frac{1+\varepsilon}{\varepsilon} \right)^{\frac{1}{m}} \cdot V\left(\frac{\tau}{d_{\min}\cdot \left(\frac{1+\varepsilon}{\varepsilon} \right)^{\frac{1}{m}}} \right) \\
& \overset{Eq.(\ref{eq:Vx_monotony})}{\leq} d_{min}^{new}\cdot V\left(\frac{\tau}{d_{\min}^{new}}\right)
\end{align*}
For $h\in\left\{1,\ldots , m-1\right\}$, there is
\begin{align}
\nonumber
V_{min}^{new}(\tau) & \overset{Eq.(\ref{eq:Vnewmin_tight})}{\leq} d_{min}\cdot  \left(\frac{1+\varepsilon}{\varepsilon} \right)^{\frac{1}{m}} \cdot \sum\limits_{i=1}^{h} \left(\frac{\varepsilon}{1+\varepsilon}\right)^{\frac{m-i+1}{m}} + \\ \nonumber
& \hspace{40pt}  + d_{min} \cdot 
y\cdot (m-h)\cdot \left(\frac{\varepsilon}{1+\varepsilon} \right)^{\frac{m-h}{m}} \\ \nonumber
& \leq d_{min}\cdot  \left(\frac{1+\varepsilon}{\varepsilon} \right)^{\frac{1}{m}} \cdot \left(\sum\limits_{i=0}^{h-1} \left(\frac{\varepsilon}{1+\varepsilon}\right)^{\frac{m-i}{m}} + \right. \\ \nonumber
& \hspace{40pt} + \left. y\cdot (m-h)\cdot \left(\frac{\varepsilon}{1+\varepsilon} \right)^{\frac{m-(h-1)}{m}}\right)\\ \nonumber
& \overset{Eq.(\ref{eq:Vx_definition})}{\leq} d_{min}\cdot  \left(\frac{1+\varepsilon}{\varepsilon} \right)^{\frac{1}{m}}\cdot V\left(\frac{\tau}{d_{min}\cdot  \left(\frac{1+\varepsilon}{\varepsilon} \right)^{\frac{1}{m}}}\right) \\ 
\label{eq:Vnewmin_calc}
& \overset{Eq.(\ref{eq:Vx_monotony}),(\ref{eq:dmin_tight})}{\leq} d_{min}^{new}\cdot V\left(\frac{\tau}{d_{min}^{new}}\right).
\end{align}
We compare this tight slack scenario to the acceptance of a job with $d_j=d_{min}$ but without tight slack ($p_j<\frac{d_{min}}{1+\varepsilon}$). We indicate the variables of the non-tight scenario with $'$ and obtain for $d_j-p_j < \tau \leq d_{min}$
\begin{align}
\label{eq:V_min_nt}
V_{min}^{new}(\tau)-\left(V_{min}^{new}(\tau)\right)' &= \frac{d_{min}}{1+\varepsilon}-p_j.
\end{align} 
This result leads to 
\begin{align*}
\left(V_{min}^{new}(\tau)\right)' &\overset{Eq.(\ref{eq:V_min_nt})}{\leq} \frac{\left(V_{min}^{new}(d_{min})\right)'+V_{\Delta}}{V_{min}^{new}(d_{min})+V_{\Delta}}\cdot V_{min}^{new}(\tau) \\
& \overset{Eq.(\ref{eq:dmin_tight})}{\leq} \frac{\left(V_{min}^{new}(d_{min})\right)'+V_{\Delta}}{d_{min}\cdot \left(\frac{1+\varepsilon}{\varepsilon} \right)^{\frac{1}{m}}\cdot f(m,\varepsilon)} \cdot V_{min}^{new}(\tau) \\
& \overset{Eq.(\ref{eq:Vnewmin_calc})}{\leq} \frac{\left(V_{min}^{new}(d_{min})\right)'+V_{\Delta}}{f(m,\varepsilon)}\cdot V\left(\frac{\tau}{d_{min}\cdot \left(\frac{1+\varepsilon}{\varepsilon} \right)^{\frac{1}{m}}} \right)\\
& \overset{Eq.(\ref{eq:Vx_value}),(\ref{eq:dmin_tight})}{\leq} \frac{\left(V_{min}^{new}(d_{min})\right)'+V_{\Delta}}{f(m,\varepsilon)}\cdot \\
& \hspace{40pt} \cdot V\left(\frac{\tau}{\frac{\left(V_{min}^{new}(d_{min})\right)'+V_{\Delta}}{f(m,\varepsilon)}}\right) \\
& \overset{Eq.(\ref{eq:Vx_monotony})}{\leq} \left(d_{min}^{new}\right)'\cdot V\left(\frac{\tau}{\left(d_{min}^{new}\right)'}\right).
\end{align*}
For the general case $p_j\leq \frac{d_{min}}{1+\varepsilon}$, we consider next an arbitrary $\tau$ with $d_j=d_{min}=\frac{V_{min}(d_{min})+V_{\Delta}}{f(m,\varepsilon)} < \tau < d_{min}^{new}\cdot \left(\frac{\varepsilon}{1+\varepsilon}\right)^{\frac{1}{m}}$. Due to Line 7 of Algorithm~\ref{alg:online}, we have 
\begin{align*}
V_{min}^{new}(\tau)& < V_{min}^{new}(d_{min}) + \left(\tau-d_{min}\right)\cdot f(m,\varepsilon)\\
& \overset{Eq.(\ref{eq:dmin_tight})}{<} d_{min}\cdot f(m,\varepsilon)\cdot \left(\frac{1+\varepsilon}{\varepsilon}\right)^{\frac{1}{m}}  -V_{\Delta}+ \\
& \hspace{40pt} + \left(\tau-d_{min}\right)\cdot f(m,\varepsilon)\\
& \overset{Eq.(\ref{eq:Vx_value})}{<} \left(\tau +d_{min}\cdot \left(\left(\frac{1+\varepsilon}{\varepsilon}\right)^{\frac{1}{m}}-1\right)\right)\cdot  \\
& \hspace{40pt} \cdot V\left(\left(\frac{\varepsilon}{1+\varepsilon}\right)^{\frac{1}{m}} \right)-V_{\Delta}\\
& < \tau \cdot \left(\frac{1+\varepsilon}{\varepsilon}\right)^{\frac{1}{m}}\cdot V\left(\frac{\tau}{\tau\cdot\left(\frac{1+\varepsilon}{\varepsilon}\right)^{\frac{1}{m}}} \right)-V_{\Delta} \\
& \overset{Eq.(\ref{eq:Vx_monotony})}{\leq} d_{min}^{new} \cdot V\left(\frac{\tau}{d_{min}^{new}} \right)-V_{\Delta}.
\end{align*}
Therefore, Eq.~(\ref{eq:V_cond_bottom}) remains valid after the acceptance of job $J_j$. 

Finally, we allow $d_j> d_{min}$ and address $d_j-p_j < \tau< d_{min}^{new}$. Due to the range of $\tau$, we can assume $d_j\leq d_{min}^{new}$. For $V_{min}(d_{min}^{new})\geq V_{min}(d_{min})+(d_j-d_{min})\cdot f(m,\varepsilon)$, we transfer processing time in direction of earlier deadlines to obtain    
\begin{align*}
\left(V_{min}(\tau)\right)' & = \left\lbrace
\begin{array}{l}
V_{min}(\tau) \mbox{ for } 0 \leq \tau < d_j \\
\max\left\lbrace d_j\cdot f(m,\varepsilon)-V_{\Delta}, V_{min}(\tau) \right\rbrace \mbox{ for } d_j\leq \tau
\end{array}
\right.
\end{align*}
Again we use the notation $(\ldots )'$ to indicate values after the transformation. Since we do not generate a schedule after the transformation, the missing continuity of $\left( V_{min}(\tau)\right)'$ for $\tau=d_j$ does not matter. The transformation leads to $\left( d_{min}\right)'=d_j$. Due to Eq.~(\ref{eq:Vx_monotony}), Eq.~(\ref{eq:V_cond_bottom}) remains valid after the transformation and we have already shown the validity of Eq.~(\ref{eq:V_cond_bottom}) after the acceptance of $J_j$. This validity also holds for the original values of $V_{min}^{new}(\tau)\leq \left(V_{min}^{new}(\tau)\right)'$.

For $V_{min}(d_{min}^{new})< V_{min}(d_{min})+(d_j-d_{min})\cdot f(m,\varepsilon)$, we determine the largest value $\delta \in (\max\{d_{min},d_j-p_j\},d_j)$ with $V_{min}(d_{min}^{new}) +(d_j-\delta)=V_{min}(d_{min})+(\delta -d_{min})\cdot f(m,\varepsilon)$. $\delta$ exists due to 
\begin{align*}
V_{min}(d_{min}^{new})+p_j& = (d_{min}^{new}-d_{min})\cdot f(m,\varepsilon)+V_{min}(d_{min}) \\
& > \left(\max\{d_{min},d_j-p_j\}-d_{min}\right)\cdot f(m,\varepsilon)+ \\
& \hspace{40pt} + V_{min}(d_{min}).
\end{align*}
We split job $J_j$ into two jobs $J_{j_1}$ and $J_{j_2}$ with $d_{j_1}=d_j$, $p_{j_1}=d_j-\delta$, $d_{j_2}=\delta$, and $p_{j_2}=p_j-(d_j-\delta )$. Both jobs observe the slack condition. We first submit and accept $J_{j_1}$. The acceptance of $J_{j_1}$ does not increase $d_{min}$  since $V_{min}(\tau)$ does not increase for any $\tau\leq \delta$. For any $\tau\in (\delta, d_{min}^{new}]$, we have 
\begin{align*}
V_{min}(\tau)-V_{min}(d_{min})+p_{j_1} & \leq V_{min}(d_{min}^{new})-V_{min}(d_{min})+ \\
& \hspace{40pt} +d_j-\delta \\
& \leq (\delta-d_{min})\cdot f(m,\varepsilon) \\
& < (\tau-d_{min})\cdot f(m,\varepsilon).
\end{align*}
Therefore, the acceptance does not change the validity of Eq.~(\ref{eq:V_cond_top}). For $J_{j_2}$, the previously discussed case applies. 
\qed 
\end{proof}
Due to Lemma~\ref{lem:valid}, Algorithms~\ref{alg:online} and \ref{alg:schedule} guarantee $V_{min}(\tau)|_t\leq (\tau-t)\cdot m$. Therefore, there is a valid preemptive schedule for the accepted jobs due to Theorem~\ref{thm:horn}, and Algorithm~\ref{alg:schedule} generates a valid schedule.

We now determine the performance of Algorithms~\ref{alg:online} and \ref{alg:schedule}. We begin with introducing a property of the generated schedule if all jobs have the same submission time. 
\begin{lemma}
\label{lem:perf_sched}
Consider a job set ${\mathcal{J}}$ with the common submission time $0$ for all jobs. If there is a valid schedule for all these jobs then the repeated execution of Algorithm~\ref{alg:schedule} generates a schedule $S$ for interval $[0,t)$ with 
\begin{align}
\label{eq:sparse}
P^*[0,t)({\mathcal{J}})-P[0,t)(S,{\mathcal{J}})\leq \frac{1}{4}\cdot P^*[0,t)({\mathcal{J}})
\end{align}
for any time $t>0$.
\end{lemma}

\begin{proof}
If no machine is idle in schedule $S$ then the claim clearly holds. Therefore, we assume time $\tau<t$ to be the first time in schedule $S$ with at least one machine being idle. Then  all jobs completing after time $\tau$ in schedule $S$ neither can start nor are preempted after time $\tau$. Let $k<m$ be the number of the jobs completing after time $t$ in $S$. Algorithm~\ref{alg:schedule} produces $P[0,t)(S,{\mathcal{J}})\geq m\cdot \tau +k\cdot (t-\tau)$.   

Since a job can contribute at most processing time $t$ to interval $[0,t)$ in any schedule, we have 
\begin{align*}
P^*[0,t)({\mathcal{J}})-P[0,t)(S,{\mathcal{J}})&\leq k\cdot \tau
\end{align*} 
resulting in
\begin{align*}
\frac{P^*[0,t)({\mathcal{J}})-P[0,t)(S,{\mathcal{J}})}{P^*[0,t)({\mathcal{J}})} & \leq \frac{k\cdot \tau}{P[0,t)(S,{\mathcal{J}})+k\cdot \tau} \\
& \leq \frac{k\cdot \tau}{m\cdot \tau +k\cdot (t-\tau)+k\cdot \tau} \\
& \leq \frac{k\cdot \tau}{m\cdot \tau +k\cdot t}\leq \frac{1}{4}.
\end{align*}
Due to inequalities $0<k< m$ and $0\leq \tau\leq t$, $k=0.5 \cdot m$ and $\tau=0.5\cdot t$ produce a maximum for $k\cdot \tau/(m\cdot \tau+k\cdot t)$. 
\qed
\end{proof}

Only Line 3 or Line 7 of Algorithm~\ref{alg:online} may lead to an increase of  $d_{min}$. We say an increase of $d_{min}$ is a \textit{time progression increase} if it occurs due to Line 3. Any other increase of $d_{min}|_t$ from time $\tau$ to time $\tau'$ does not involve progression of time and requires $V_{min}^{new}(\tau')|_t-V_{min}(\tau)|_t=(\tau'-\tau)\cdot f(m,\varepsilon)$. We use this distinction of increases of $d_{min}$ to prove the competitive ratio for Algorithms~\ref{alg:online} and \ref{alg:schedule}.
\begin{theorem}
\label{thm:cr_preemptive}
The $P_m|\varepsilon,\mbox{online},pmtn|\sum p_j\cdot(1-U_j)$ problem 
admits a deterministic online algorithm with competitive ratio at most
\begin{align*}
\frac{m}{f(m,\varepsilon)}& = m\cdot (1+\varepsilon)\cdot \left(\left(\frac{1+\varepsilon}{\varepsilon} \right)^{\frac{1}{m}}-1 \right).
\end{align*}
\end{theorem}

\begin{proof}
We partition schedule $S$ generated by Algorithms~\ref{alg:online} and \ref{alg:schedule} into an alternating sequence of \textit{rejection} and \textit{non-rejection} intervals. A non-rejection interval starts with some $d_{min}$ value such that every increase of $d_{min}$ in this interval is a time progression increase. At its end, $d_{min}|_t=t$ holds.  We combine any neighboring non-rejection intervals into a single non-rejection interval. Algorithm~\ref{alg:online} does not reject any job with either its submission time or its deadline being within this interval. An interval separating two subsequent non-rejection intervals is a rejection interval and does not contain any time progression increase of $d_{min}$.

Schedule $S$ will either start with a rejection interval or with a non-rejection interval but it will always finish with a non-rejection interval. Due to Lemma~\ref{lem:perf_sched}, we have for any non-rejection interval $[t,t')$:
\begin{align*}
\frac{P^*[t,t')({\mathcal{J}})}{P[t,t')(S,{\mathcal{J}})} &\overset{Eq.(\ref{eq:sparse})}{\leq} \frac{4}{3}.
\end{align*}
For any rejection interval $[t,t')$, the increase of the $d_{min}$ guarantees:
\begin{align*}
\frac{P^*[t,t')({\mathcal{J}})}{P[t,t')(S,{\mathcal{J}})} &\leq \frac{ (t'-t)\cdot m}{(t'-t)\cdot f(m,\varepsilon)}= \frac{m}{f(m,\varepsilon)}.
\end{align*}
For $\varepsilon\leq 1$, we have $\frac{m}{f(m,\varepsilon)}\geq \lim_{m\rightarrow \infty} \frac{m}{f(m,1)}= 2\cdot \ln(2)>\frac{4}{3}$.
\qed
\end{proof}
Note that for $m=1$, Algorithms~\ref{alg:online} and \ref{alg:schedule} guarantee the same competitive ratio $\frac{1+\varepsilon}{\varepsilon}$ as the optimal greedy algorithm for a single machine.

\subsection*{Lower Bounds}
In this subsection, we provide a lower bound of the competitive ratio for the $P_m|\varepsilon,\mbox{online},pmtn|\sum p_j\cdot(1-U_j)$ problem.  First we describe the general concept. Our adversary controls the job sequence $\mathcal{J}$ that only uses jobs with release date $0$ and consists of up to $m+2$ blocks. In blocks 1 to $m+1$, the adversary submits identical jobs up to a maximum number and proceeds with the next block as soon as we have accepted a block specific target number of these jobs. If we do not accept this target number then the adversary stops after submitting the above mentioned maximum number of jobs of this block. In block $m+2$, the adversary unconditionally submits the corresponding maximum number of jobs and stops afterwards.  

The job of the first block has deadline $d_1=1+\varepsilon$ and an arbitrarily small positive processing time $p_1=\delta\ll 1$ such that the target number $\frac{\varepsilon}{\delta}\cdot\sum_{i=0}^{m-1} \left(\frac{1+\varepsilon}{\varepsilon}\right)^{\frac{i}{m}}$ is an integer. The maximum number of jobs is $\lfloor \frac{m\cdot (1+\varepsilon)}{\delta}\rfloor$. 

The next $m$ blocks are similar. The job of block $k$ with $2\leq k \leq m+1$ has a tight slack with processing time $p_k=\left(\frac{1+\varepsilon}{\varepsilon}\right)^{\frac{k-2}{m}}$. For these blocks, the target number is 1 while the maximum number is $\lfloor m\cdot (1+\varepsilon) \rfloor$. 

The job of the final block $m+2$ has a tight slack with processing time $p_{m+2}=\frac{1+\varepsilon}{\varepsilon}\cdot(1-\delta)$. The maximum number is $\lfloor m\cdot (1+\varepsilon) \rfloor$ as well.

Assume that we have accepted a total processing time of at most $P_{acc}$ and that we can arrange the submitted jobs to obtain a schedule with a total processing time of at least $P_{ideal}-const\cdot \delta$ then a lower bound of the competitive ratio for this scenario is $P_{ideal}/P_{acc}$ since $\delta$ is arbitrarily small. 

\begin{theorem}
\label{thm:lbound_prmp}
The competitive ratio for any deterministic online algorithm for the $P_m|\varepsilon,\mbox{online},\mbox{pmtn}|\sum p_j\cdot(1-U_j)$ problem is at least 
\begin{align*}
\frac{\lfloor m\cdot (1+\varepsilon)\rfloor}{\varepsilon\cdot\sum\limits_{i=0}^{m-1}\left(\frac{1+\varepsilon}{\varepsilon}\right)^{\frac{i}{m}}}= \lfloor m\cdot (1+\varepsilon)\rfloor \cdot \left(\left(\frac{1+\varepsilon}{\varepsilon}\right)^{\frac{1}{m}}-1\right)\\
\end{align*}
\end{theorem} 

\begin{proof}
If the adversary stops before starting block 2 then the submitted total processing time is more than $m\cdot (1+\varepsilon)-\delta$ while the algorithm has accepted a total processing time of at most $\varepsilon\cdot\sum_{i=0}^{m-1} \left(\frac{1+\varepsilon}{\varepsilon}\right)^{\frac{i}{m}}-\delta$ resulting in a lower bound of 
\begin{align*}
\frac{m\cdot (1+\varepsilon)}{\varepsilon\cdot\sum\limits_{i=0}^{m-1} \left(\frac{1+\varepsilon}{\varepsilon}\right)^{\frac{i}{m}}}& = m\cdot (1+\varepsilon) \cdot \left(\left(\frac{1+\varepsilon}{\varepsilon}\right)^{\frac{1}{m}}-1\right).
\end{align*}
 
If the adversary stops before starting block $k+1$ with $2\leq k\leq m+1$ then the algorithm has accepted a total processing time of $p_k\cdot \varepsilon\cdot \sum_{i=0}^{m-1}\left(\frac{1+\varepsilon}{\varepsilon}\right)^{\frac{i}{m}}$ due to $p_2=1$ and  
\begin{align}
\nonumber
P_h & = p_h\cdot \varepsilon\cdot \sum_{i=0}^{m-1}\left(\frac{1+\varepsilon}{\varepsilon}\right)^{\frac{i}{m}}+p_{h} \\ \nonumber
&= p_h\cdot\varepsilon\cdot \left(\sum\limits_{i=0}^{m-1}\left(\frac{1+\varepsilon}{\varepsilon}\right)^{\frac{i}{m}}+\frac{1}{\varepsilon}\right) \\ \nonumber
& \overset{Eq.(\ref{eq:func_add})}{=}p_h\cdot\varepsilon\cdot\sum\limits_{i=1}^{m}\left(\frac{1+\varepsilon}{\varepsilon}\right)^{\frac{i}{m}}\\
\label{eq:lb_intermediate}
&=p_h\cdot \left(\frac{1+\varepsilon}{\varepsilon}\right)^{\frac{1}{m}}\cdot \varepsilon\cdot \sum_{i=0}^{m-1}\left(\frac{1+\varepsilon}{\varepsilon}\right)^{\frac{i}{m}} \\ \nonumber  
 & =p_{h+1}\cdot \varepsilon\cdot \sum_{i=0}^{m-1}\left(\frac{1+\varepsilon}{\varepsilon}\right)^{\frac{i}{m}}  \mbox{ for }2\leq h<k.
\end{align}

Due to Theorem~\ref{thm:horn}, we can use  at least all $\lfloor m\cdot(1+\varepsilon) \rfloor$ jobs with processing time $p_k$ in the optimal schedule. This usage results in a lower bound of 
\begin{align*}
\frac{p_k\cdot \lfloor m\cdot (1+\varepsilon)\rfloor}{p_k\cdot \varepsilon\cdot\sum\limits_{i=0}^{m-1} \left(\frac{1+\varepsilon}{\varepsilon}\right)^{\frac{i}{m}}}& =  \frac{\lfloor m\cdot (1+\varepsilon)\rfloor}{\varepsilon\cdot\sum\limits_{i=0}^{m-1} \left(\frac{1+\varepsilon}{\varepsilon}\right)^{\frac{i}{m}}}\\
&= \lfloor m\cdot (1+\varepsilon)\rfloor \cdot \left(\left(\frac{1+\varepsilon}{\varepsilon}\right)^{\frac{1}{m}}-1\right).
\end{align*}
 
If the adversary starts with block $m+2$ then the algorithm has accepted exactly one job with deadline $d_k$ for $3\leq k\leq m+1$ while all other jobs have deadline $d_1=d_2$. The total accepted processing time is
\begin{align*}
P_{acc} & = p_{m+1}\cdot \varepsilon\cdot \sum_{i=0}^{m-1}\left(\frac{1+\varepsilon}{\varepsilon}\right)^{\frac{i}{m}}+p_{m+1}\\ 
& \overset{Eq.(\ref{eq:lb_intermediate})}{=}p_{m+1}\cdot \left(\frac{1+\varepsilon}{\varepsilon}\right)^{\frac{1}{m}}\cdot \varepsilon \cdot \sum_{i=0}^{m-1}\left(\frac{1+\varepsilon}{\varepsilon}\right)^{\frac{i}{m}}  \\
& = \left(\frac{1+\varepsilon}{\varepsilon}\right)^{\frac{m-1}{m}}\cdot \varepsilon \cdot \sum_{i=1}^{m}\left(\frac{1+\varepsilon}{\varepsilon}\right)^{\frac{i}{m}} \\
& = (1+\varepsilon)\cdot \sum_{i=1}^{m}\left(\frac{1+\varepsilon}{\varepsilon}\right)^{\frac{i-1}{m}} = \sum_{i=1}^{m}d_{i+1}
\end{align*}
resulting in no idle time on any machine in interval $[0,d_1)$. Since any job of block $m+2$ cannot start later than time 
\begin{align*}
d_{m+2}-p_{m+2} & =  \varepsilon\cdot p_{m+2}=  \varepsilon\cdot \left(\frac{1+\varepsilon}{\varepsilon}\cdot(1-\delta)\right) \\
&=(1+\varepsilon)\cdot (1-\delta)<d_1,
\end{align*}
we cannot accept any job of block $m+2$. 

Again, the total processing time of the optimal schedule is at least $\lfloor m\cdot(1+\varepsilon) \rfloor\cdot p_{m+2}$ resulting in a lower bound of 
\begin{align*}
 \frac{p_{m+2}\cdot \lfloor m\cdot (1+\varepsilon)\rfloor}{(1+\varepsilon) \cdot \sum\limits_{i=1}^{m}\left(\frac{1+\varepsilon}{\varepsilon}\right)^{\frac{i-1}{m}}}& =  \frac{\frac{1+\varepsilon}{\varepsilon}\cdot(1-\delta)\lfloor m\cdot (1+\varepsilon)\rfloor}{(1+\varepsilon)\cdot\sum\limits_{i=0}^{m-1} \left(\frac{1+\varepsilon}{\varepsilon}\right)^{\frac{i}{m}}} \\
& = \frac{\lfloor m\cdot (1+\varepsilon)\rfloor}{\varepsilon\cdot\sum\limits_{i=0}^{m-1} \left(\frac{1+\varepsilon}{\varepsilon}\right)^{\frac{i}{m}}} \\
& \overset{\delta\rightarrow 0}{=} \lfloor m\cdot (1+\varepsilon)\rfloor \cdot \left(\left(\frac{1+\varepsilon}{\varepsilon}\right)^{\frac{1}{m}}-1\right)
\end{align*}
\qed
\end{proof}
Theorems~\ref{thm:lbound_prmp} and \ref{thm:cr_preemptive} show that Algorithms~ \ref{alg:online} and \ref{alg:schedule} guarantee a tight (for $m\cdot (1+\varepsilon)$ being an integer) or an almost tight competitive ratio for $\varepsilon\leq 1$. 

We can increase the lower bound to 
\begin{align*}
\max\left\{m\cdot (1+\varepsilon), \lfloor m\cdot (1+\varepsilon) \rfloor + \frac{\varepsilon}{1+\varepsilon}\cdot \left(\frac{1+\varepsilon}{\varepsilon}\right)^{\frac{1}{m}}\right\}
\end{align*}  
by adding accepted jobs with deadline $d_1$ in the optimal schedule if there are idle machines. The new lower bound reduces the gap for non-integral values of $m\cdot (1+\varepsilon)$ but cannot guarantee a tight lower bound for all those values. 

Algorithms~\ref{alg:online} and \ref{alg:schedule} guarantee a better competitive ratio than the known greedy approach for all values of $\varepsilon$ and $m>1$. In comparison to the single machine problem, the competitive factor is significantly better for small values of $\varepsilon$. Although we focus on small values of $\varepsilon$, we like to remark that for the parallel problem, the competitive ratio cannot approach 1 for large values of $\varepsilon$ due to the scheduling of the jobs (see Algorithm~\ref{alg:schedule}) while the competitive ratio tends to 1 for $\varepsilon\rightarrow \infty$ in the single machine problem.  

\section{Online Scheduling without Preemption}
\label{sec:nonpreemptive}

In this section, we present an online algorithm for the non-preemptive version. Our non-preemptive Algorithm~\ref{alg:nonpreemptive} Online Allocation also uses a deadline threshold and determines this threshold based on an exponential sequence of processing times on the machines. Remember that in our online model, the algorithm must decide immediately after submission of a job whether to accept or reject the job. If we accept the job then we also must immediately allocate the job to a specific machine with a specific starting time and cannot change this allocation later. Once we have allocated a job to a machine, we start this job as early as possible on this machine. Therefore, it is not necessary to consider the job schedule but only the not yet completed load $l(m_i)$ on each machine $m_i$. We index the machines in decreasing order of their loads: $l(m_1)\geq l(m_2)\geq \ldots \geq l(m_{m-1})\geq l(m_m)$. We use the notation $l(m_i)|_{t}$ to describe the value $l(m_i)$ at a specific time $t$ although we are aware that there may be different load values at time $t$ due to the acceptance of jobs. For disambiguation, we apply a similar notation $l(m_i)|_{J_j}$ to indicate the value $l(m_i)$ when making the decision to accept job $J_j$. $l(m_i)|_{J_j}'$ describes the update of the load $l(m_i)$ after accepting job $J_j$. Note that the allocation of $J_j$ to a machine may result in a re-indexing of the machines. This re-indexing is already considered in $l(m_i)|_{J_j}'$. 

The next two expressions specify our deadline threshold:
\begin{align}
\label{eq:def_dilim}
d^{i}_{lim}|_{t} & =l(m_i)|_{t}\left(\frac{1+\varepsilon}{\varepsilon}\right)^{\frac{i}{m}}+t  \\
\label{eq:def_dlim} 
d_{lim}|_t & =\max_{1\leq i\leq m}d^{i}_{lim}|_t 
\end{align}
If $d_{lim}|_t=d^{i}_{lim}|_t$ holds then we say that machine $m_i$ determines $d_{lim}$. We use the notation $d_{lim}|_{J_j}'$ to denote the value of $d_{lim}|_{J_j}$ after the acceptance of job $J_j$.   

In Algorithm~\ref{alg:nonpreemptive}, Line 1, we initially  set $d_{lim}$ to 0. At the submission of a new job $J_j$, we first update $d_{lim}$ to consider any progression of time since the previous submission (Line 3). If the resulting $d_{lim}|_{J_j}>d_j$ holds then we reject $J_j$. Otherwise we accept $J_j$ and allocate it to the machine that produces the smallest value $d_{lim}|_{J_j}'$ (Lines 7 to 9). Contrary to Algorithm~\ref{alg:nonpreemptive}, Algorithm~\ref{alg:online} allows a decrease of $d_{lim}$ due to progression of time. Therefore, there is no need for any compensation value.  

\begin{algorithm}[ht]
\caption{Online Allocation}
\label{alg:nonpreemptive}
\begin{algorithmic}[1]
\STATE{$d_{lim}=0$}
\FOR{each newly submitted job $J_j$}
\STATE{update $d_{lim}$}
\IF{$d_j< d_{lim}$}
\STATE{reject $J_j$}
\ELSE
\STATE{accept $J_j$}
\STATE{allocate $J_j$ to a machine producing minimal $d_{lim}$}
\STATE{update $d_{lim}$}
\ENDIF
\ENDFOR
\end{algorithmic}
\end{algorithm}

\begin{lemma}
\label{lem:legal_allocation}
There is a legal schedule for all jobs accepted by Algorithm~\ref{alg:nonpreemptive}.
\end{lemma}

\begin{proof}
Acceptance of job $J_j$ requires $d_j\geq d_{lim}|_{J_j}$. If $J_j$ is allocated to machine $m_m$ then it completes at 
\begin{align*}
l(m_m)|_{J_j}+r_j+p_j &\overset{Eq.(\ref{eq:def_dilim})}{=}(d^m_{lim}|_{r_j}-r_j)\cdot\frac{\varepsilon}{1+\varepsilon}+r_j+p_j \\
& \overset{Eq.(\ref{eq:def_dlim})}{\leq}(d_{lim}|_{r_j}-r_j)\cdot\frac{\varepsilon}{1+\varepsilon}+r_j+p_j\\
& \leq (d_j-r_j)\cdot \frac{\varepsilon}{1+\varepsilon} +r_j+ \frac{d_j-r_j}{1+\varepsilon} \leq  d_j.
\end{align*}
\qed
\end{proof}
For the proof of the competitive factor, we first establish a property of the schedule.
\begin{lemma}
\label{lem:property_allocation}
Algorithm~\ref{alg:nonpreemptive} always guarantees 
\begin{align}
\label{eq:sum_lm}
l(m_1)|_{t}+l(m_2)|_{t} & \geq (d_{lim}|_{t}-t)\cdot \left(\frac{\varepsilon}{1+\varepsilon} \right)^{\frac{1}{m}}.
\end{align}
\end{lemma}

\begin{proof}
The claim holds whenever there is no load on any machine or if machine $m_1$ determines $d_{lim}$ due to Eq.~(\ref{eq:def_dilim}). Similar to Lemma~\ref{lem:valid}, we distinguish the events progression of time and acceptance of a new job. 

First we consider progression from time $t$ to time $t'$ and assume that the claim holds at time $t$. For every machine $m_i$ with $l(m_i)|_t\leq t'$, it leads to $l(m_i)|_{t'}=0\overset{Eq.(\ref{eq:def_dilim})}{\Rightarrow} d^{i}_{lim}|_{t'}=t'$. For all other machines, we have  
\begin{align*}
d^{i}_{lim}|_{t'}-t'& \overset{Eq.(\ref{eq:def_dilim})}{=}l(m_i)|_{t'} \cdot \left(\frac{1+\varepsilon}{\varepsilon} \right)^{\frac{i}{m}}\\
& =\left(l(m_i)|_{t}-(t'-t)\right)\cdot \left(\frac{1+\varepsilon}{\varepsilon} \right)^{\frac{i}{m}}\\
& = l(m_i)|_{t}\cdot\left(\frac{1+\varepsilon}{\varepsilon} \right)^{\frac{i}{m}}-(t'-t)\cdot\left(\frac{1+\varepsilon}{\varepsilon} \right)^{\frac{i}{m}} \\
& \overset{Eq.(\ref{eq:def_dilim})}{=} (d^{i}_{lim}|_{t}-t)-(t'-t)\cdot\left(\frac{1+\varepsilon}{\varepsilon} \right)^{\frac{i}{m}}.
\end{align*}
The function $d^{i}_{lim}|_{t'}-t'-\left(d^{i}_{lim}|_t-t\right)=-(t'-t)\cdot\left(\frac{1+\varepsilon}{\varepsilon} \right)^{\frac{i}{m}}$ is decreasing in $i$. Therefore, progression of time does not require a change in the machine order. In addition, machine $m_1$ determines $d_{lim}|_{t'}$ if it has already determined $d_{lim}|_t$. If machine $m_1$ does not determine $d_{lim}|_{t'}$ then we apply induction and assume for all $i>1$ with $l(m_i)|_{t}>t'-t$
\begin{align*}
d^{i}_{lim}|_{t}-t & \leq \left(l(m_1)|_{t}+l(m_2)|_{t}\right) \cdot \left(\frac{1+\varepsilon}{\varepsilon} \right)^{\frac{1}{m}}. 
\end{align*}
Then Eq.(\ref{eq:def_dilim}) yields
\begin{align*}
\frac{d^{i}_{lim}|_{t}-t}{l(m_i)|_{t}}\cdot l(m_i)|_{t'} & \leq \frac{l(m_1)|_t+l(m_2)|_t}{l(m_i)|_t}\cdot l(m_i)|_{t'}\cdot \\
& \hspace{40pt} \cdot \left(\frac{1+\varepsilon}{\varepsilon} \right)^{\frac{1}{m}}\\
 \left(\frac{1+\varepsilon}{\varepsilon} \right)^{\frac{i}{m}}\cdot l(m_i)|_{t'} &\leq \left(l(m_1)|_t+l(m_2)|_t\right)\cdot \left(1 - \frac{t'-t}{l(m_i)|_t}\right) \cdot \\
 & \hspace{40pt} \cdot \left(\frac{1+\varepsilon}{\varepsilon} \right)^{\frac{1}{m}}\\
 d^{i}_{lim}|_{t'} -t'& \leq  \left(l(m_1)|_t+l(m_2)|_t- 2 (t'-t)\right) \cdot \\
 & \hspace{40pt} \cdot \left(\frac{1+ \varepsilon}{\varepsilon} \right)^{\frac{1}{m}} \\
& \leq \left(l(m_1)|_{t'}+l(m_2)|_{t'}\right) \cdot \left(\frac{1+\varepsilon}{\varepsilon} \right)^{\frac{1}{m}}
\end{align*}
resulting in   
\begin{align*}
d_{lim}|_{t'}-t' & \overset{Eq.(\ref{eq:def_dlim})}{=} \max_{1\leq i \leq m}d^{i}_{lim}|_{t'} -t' \\
& \leq  \left(l(m_1)|_{t'}+l(m_2)|_{t'}\right) \cdot \left(\frac{1+\varepsilon}{\varepsilon} \right)^{\frac{1}{m}}.
\end{align*}
Next, we assume that Algorithm~\ref{alg:nonpreemptive} accepts job $J_j$ and that the claim holds before accepting $J_j$. Algorithm~\ref{alg:nonpreemptive} allocates $J_j$ to machine $m_k$. Re-indexing puts this machine at position $i\leq k$. For $k=1$, we either have $d_{lim}|_{J_j}=d_{lim}|_{J_j}'$ or $d^1_{lim}|_{J_j}'=d_{lim}|_{J_j}'$. We have already stated that the claim holds in the latter case while in the former case, the claim continues to hold since no machine load decreases. Therefore, we assume $k>1$ and $d^1_{lim}|_{J_j}'<d_{lim}|_{J_j}'$, that is, machine $m_h$ with $h>1$ determines $d_{lim}|_{J_j}'$.  

For $l(m_1)|_{J_j}+p_j+r_j\leq d_j$, Line 8 of Algorithm~\ref{alg:nonpreemptive} guarantees  
\begin{align*}
l(m_1)|_{J_j}'+l(m_2)|_{J_j}' & \geq l(m_1)|_{J_j}+p_j \\
& >(d_{lim}|_{J_j}'-r_j) \cdot \left(\frac{\varepsilon}{1+\varepsilon} \right)^{\frac{1}{m}}.
\end{align*}
For $l(m_1)|_{J_j}+p_j+r_j> d_j\geq d_{lim}|_{J_j}$, we have either $d_{lim}|_{J_j}=d_{lim}|_{J_j}'$ and the claim continues to hold or $i\leq h \leq k$ with $l(m_h)|_{J_j}'\leq l(m_{h-1})|_{J_j}$ and  
\begin{align*}
d_{lim}|_{J_j}'-r_j& =d^{h}_{lim}|_{J_j}'-r_j\overset{Eq.(\ref{eq:def_dilim})}{=}l(m_h)|_{J_j}'\cdot \left(\frac{1+\varepsilon}{\varepsilon} \right)^{\frac{h}{m}} \\
& \leq l(m_{h-1})|_{J_j}\cdot \left(\frac{1+\varepsilon}{\varepsilon} \right)^{\frac{h-1}{m}}\cdot \left(\frac{1+\varepsilon}{\varepsilon} \right)^{\frac{1}{m}}\\
& \overset{Eq.(\ref{eq:def_dlim})}{\leq} \left(d_{lim}|_{J_j}-r_j \right)\cdot \left(\frac{1+\varepsilon}{\varepsilon} \right)^{\frac{1}{m}}\\
& \leq  \left(l(m_1)|_{J_j}+p_j\right) \cdot \left(\frac{1+\varepsilon}{\varepsilon} \right)^{\frac{1}{m}} \\
& \leq \left(l(m_1)|_{J_j}'+l(m_2)|_{J_j}'\right)\cdot \left(\frac{1+\varepsilon}{\varepsilon} \right)^{\frac{1}{m}}.  
\end{align*}
\qed   
\end{proof}

Next, we prove a competitive ratio of Algorithm~\ref{alg:nonpreemptive}. To this end, we define busy and usable intervals of schedule $S$ generated by Algorithm~\ref{alg:nonpreemptive}: interval $[t,t')$ is busy if $t=r_j$ is the submission time of a job $J_j$ and we have no load on any machine after progressing to $t$ but before accepting job $J_j$. $t'$ is the first time after $t$ with no load on any machine after progressing to $t'$ or a later time. Note that the submission time of any rejected job must be within a busy interval. ${\mathcal{J}}_{r}$ is the set of all rejected jobs with submission times in interval $[t,t')$. Interval $[t,t'')$ is usable with $t''=\max\left\lbrace t', \max_{J_k\in{\mathcal{J}}_{r}}\{d_k\}\right\rbrace$. If rejected job $J_k$ determines $t''$ then we have $r_k\geq t$ and  
\begin{align}
\nonumber
t''-r_k & = d_k-r_k < d_{lim}|_{r_k}-r_k \\
\label{eq:schedule_volume}
& \overset{Eq.(\ref{eq:sum_lm})}{\leq} \left(l(m_1)|_{r_k}+l(m_2)|_{r_k}\right) \cdot \left(\frac{1+\varepsilon}{\varepsilon} \right)^{\frac{1}{m}} \\ \nonumber
& \leq \sum_{i=1}^m l(m_i)|_{r_k}\cdot \left(\frac{1+\varepsilon}{\varepsilon} \right)^{\frac{1}{m}}.
\end{align}
Using the definition of $P[t,t')(S,{\mathcal{J}})$ in Section~\ref{sec:pmtn}, we obtain 
\begin{align}
\nonumber
t''-t & = t''-r_k + r_k-t \\ \nonumber
& \overset{Eq.(\ref{eq:schedule_volume})}{\leq} P[r_k,t')(S,{\mathcal{J}})\cdot \left(\frac{1+\varepsilon}{\varepsilon} \right)^{\frac{1}{m}}+r_k-t\\
\label{eq:busyinterval}
& \leq P[t,t')(S,{\mathcal{J}})\cdot \left(\frac{1+\varepsilon}{\varepsilon} \right)^{\frac{1}{m}}.
\end{align}
Note that any time instance of schedule $S$ with at least one busy machine belongs to a busy interval  and it is not possible to execute a rejected job at a time instance that does not belong to any usable interval. Such a restriction does not exist for an accepted job. 
 
\begin{theorem}
\label{thm:nonpreemptive_online}
The $P_m|\varepsilon,\mbox{online}|\sum p_j\cdot(1-U_j)$ problem 
admits a deterministic online algorithm with competitive ratio at most 
\begin{align*}
m \cdot\left(\frac{1+\varepsilon}{\varepsilon}\right)^{\frac{1}{m}}+1.
\end{align*}
\end{theorem} 

\begin{proof}
Since no two busy intervals overlap, we assume $S$ to be a single busy interval $[t=0,t')$ with the corresponding usable interval $[t=0,t''\leq \max_{J_j\in\mathcal{J}}\{d_j\})$. Then we have 
\begin{align*}
\frac{P^*[0,t'')}{P[0,t')(S,{\mathcal{J}})}& \leq
\frac{m\cdot t''+P[0,t')(S,{\mathcal{J}})}{P[0,t')(S,{\mathcal{J}})}\\
& \leq \frac{m\cdot t''}{P[0,t')(S,{\mathcal{J}})}+1\\
& \overset{Eq.(\ref{eq:busyinterval})}{\leq} \frac{m\cdot P[0,t')(S,{\mathcal{J}})\cdot\left(\frac{1+\varepsilon}{\varepsilon}\right)^{\frac{1}{m}}}{P[0,t')(S,{\mathcal{J}})}+1\\
& \leq m\cdot \left(\frac{1+\varepsilon}{\varepsilon}\right)^{\frac{1}{m}} +1.
\end{align*}
\qed
\end{proof}
The expression $m \cdot\left(\frac{1+\varepsilon}{\varepsilon}\right)^{\frac{1}{m}}+1$ has its minimal value $e\cdot\ln{\frac{1+\varepsilon}{\varepsilon}}+1 $ 
for $m = \ln{\frac{1+\varepsilon}{\varepsilon}}$. Since it decreases for $1\leq m \leq \ln{\frac{1+\varepsilon}{\varepsilon}}$ and increases for $m \geq \ln{\frac{1+\varepsilon}{\varepsilon}}$, we may obtain a better competitive ratio by partitioning the system in groups of approximately $\ln{\frac{1+\varepsilon}{\varepsilon}}$ machines. In particular, this approach yields the following result.
\begin{corollary}
For $\ln{\frac{1+\varepsilon}{\varepsilon}}$ being an integer and $m = i\cdot \ln{\frac{1+\varepsilon}{\varepsilon}}$, $i\in \mathbb{N}$, the $P_m|\varepsilon,\mbox{online}|\sum p_j\cdot(1-U_j)$ problem 
admits a deterministic online algorithm with competitive ratio at most 
\begin{align*}
e \cdot\ln\left(\frac{1+\varepsilon}{\varepsilon}\right)+1.
\end{align*}
\end{corollary}
\begin{proof}
We partition the machines into sets of machines with size $\ln{\frac{1+\varepsilon}{\varepsilon}}$. For each set of machines, we run Algorithm~\ref{alg:nonpreemptive} and obtain a competitive ratio of at most $e\cdot\ln{\frac{1+\varepsilon}{\varepsilon}}+1$ for the set of jobs allocated to these machines. We consider any job not allocated to these machines on the next set of machines and repeat the procedure. Hence the competitive ratio of Theorem~\ref{thm:nonpreemptive_online} with $m=\ln \frac{1+\varepsilon}{\varepsilon}$ carries over.
\end{proof}

For $m=1$, Theorem~\ref{thm:nonpreemptive_online} produces the competitive ratio of the single machine case although Algorithm~\ref{alg:nonpreemptive} does not greedily accept every job that we can schedule: in particular, it rejects jobs with a deadline below the threshold deadline even if their small processing time allows a valid schedule. This rejection has no influence on the competitive ratio. 

For $m>1$ and $\frac{1+\varepsilon}{\varepsilon}\leq e \Leftrightarrow\varepsilon > 0.58$, greedy acceptance has a better competitive ratio than Algorithm~\ref{alg:nonpreemptive} while the competitive ratio of Algorithm~\ref{alg:nonpreemptive} is clearly better for smaller values of $\varepsilon$ and $m>1$. 

We also briefly point out how to use our deterministic algorithm for parallel machines to obtain a randomized algorithm for small values of $\varepsilon$.
\begin{corollary}
The $1|\varepsilon,\mbox{online}|\sum p_j\cdot(1-U_j)$ problem admits a randomized online algorithm against an oblivious adversary with competitive ratio at most
\begin{align*}
\min_{m\in \left\lbrace \left\lfloor\ln\left(\frac{1+\varepsilon}{\varepsilon}\right) \right\rfloor,\left\lceil\ln\left(\frac{1+\varepsilon}{\varepsilon}\right) \right\rceil \right\rbrace} \left\lbrace m^2 \cdot \left(\frac{1+\varepsilon}{\varepsilon}\right)^{\frac{1}{m}}+ m \right\rbrace 
\end{align*}
if $\varepsilon>0$ is sufficiently small.
\end{corollary}
\begin{proof}
We simulate our deterministic Algorithm~\ref{alg:nonpreemptive} with $m\in \left\lbrace \left\lfloor\ln\left(\frac{1+\varepsilon}{\varepsilon}\right) \right\rfloor,\left\lceil\ln\left(\frac{1+\varepsilon}{\varepsilon}\right) \right\rceil \right\rbrace$ machines. 
Denote by $ALG_{m}$ and $OPT_{m}$ the obtained utilization and the optimal utilization on these machines, respectively. Then we have 
\begin{align*}
m\cdot\left(\frac{1+\varepsilon}{\varepsilon}\right)^{\frac{1}{m}}+1& \geq  \frac{OPT_{m}}{ALG_{m}}\geq \frac{OPT}{ALG_{m}}.
\end{align*}
We now choose one of the simulated machines uniformly at random which, on expectation, has exactly an $1/m$ fraction of the utilization of the simulated schedule, that is $\mathbb{E}[ALG] = \frac{1}{m} ALG_{m}$ resulting in $m^2\cdot\left(\frac{1+\varepsilon}{\varepsilon}\right)^{\frac{1}{m}}+m \geq \frac{OPT}{\mathbb{E}[ALG]}$.
\qed
\end{proof}
For $\ln\left(\frac{1+\varepsilon}{\varepsilon}\right)$ being integral, we obtain the competitive ratio $\ln\left(\frac{1+\varepsilon}{\varepsilon}\right)\cdot \left(\ln\left(\frac{1+\varepsilon}{\varepsilon}\right)\cdot e+1 \right)$.

We note that a similar technique cannot be applied in the preemptive case, as we cannot independently treat the machines.

\subsection*{Lower Bounds}

For the $P_m|\varepsilon,\mbox{online}|\sum p_j\cdot(1-U_j)$ problem, no lower bound of the competitive ratio is known. In comparison to the single machine problems, it is more difficult to determine a lower bound for the non-preemptive problem than for the preemptive problem since any non-preemptive scheduling algorithm faces restrictions in the machine space that are not present in the preemptive problem. Looking back at the lower bound instance in the preemptive case, we realize that the generated schedules are non-preemptive. Therefore, we can use this approach to determine a lower bound although the common submission time in this approach prevents us from achieving the increase  of the lower bound in the proof of Goldwasser~\cite{Gol03} over the lower bound for the single machine preemptive problem. Hence, we cannot expect to obtain a tight lower bound with this approach. Moreover, the optimal schedule in the proof of Theorem~\ref{thm:lbound_prmp} uses $\lfloor m\cdot (1+\varepsilon) \rfloor$ jobs of a given length and tight slack to cover the interval $[0,p_j\cdot (1+\varepsilon ))$ as well as possible. For $\varepsilon<1$, we can schedule only $m$ of these jobs without applying preemption.

Therefore, we use a slightly modified approach for small $\varepsilon$. The adversary first submits a single job $J_1$ with $r_1=0$, $p_1=1$, and a sufficiently large deadline. An online algorithm must accept this job to prevent an infinite competitive ratio. We assume starting time $t$ of $J_1$. Then the adversary submits $m-1$ different groups of up to $m$ identical jobs with submission time $t$ and a tight slack such that we must accept one job of each group to prevent the target competitive ratio or an even larger competitive ratio. The adversary selects the job parameters such that each job must execute on a separate machine. After we have accepted one job of a group, the adversary continues with the submission of a job of the next group. Finally the adversary submits $m$ jobs with submission time $t$ and processing time $\frac{1}{\varepsilon}-\delta$ for an arbitrary small $\delta>0$. As in the approach of Theorem~\ref{thm:lbound_prmp}, we are not able to schedule any job of the final group.  

\begin{theorem}
\label{thm:lbound_nonprmp}
For the $P_m|\varepsilon,\mbox{online}|\sum p_j\cdot(1-U_j)$ problem with $\varepsilon< 1$, any deterministic online algorithm is at least strictly $c$-competitive with  
\begin{align*}
\frac{c}{m} &= \left( \frac{m}{(c-1)\cdot \varepsilon}\right)^{\frac{1}{m-1}}-1.
\end{align*}
\end{theorem} 

\begin{proof}
For $\varepsilon<1$, we cannot allocate any two jobs $J_j$ and $J_k$ to the same machine if both jobs have a tight slack, $r_j=r_k=t$ and $1\leq p_j,p_k< \frac{1}{\varepsilon}$. Since job $J_1$ starts at time $t$, we can include it into this group of jobs, that is, we cannot allocate any other job of this group to the same machine as job $J_1$.

Assume that we want to test whether a competitive ratio $c$ is a strict lower bound. Then we must show that there is a set of increasing processing times $p_1=1< p_2<  \ldots < p_{m}< p_{m+1}<\frac{1}{\varepsilon}$ such that we have 
\begin{align}
\label{eq:comp_ratio}
1+m\cdot p_{j} & \geq c\cdot \sum_{i=1}^{j-1}p_i\;\; \mbox{ with } \;\; 1<j\leq m+1.
\end{align}
Eq.~(\ref{eq:comp_ratio}) assumes that we have accepted exactly one job from block 1 to block $j-1$ but no job from block $j$. As already mentioned, we can only schedule the accepted jobs using exactly $j-1$ machines. In the optimal schedule, we execute job $J_1$ such that it either completes before $t$ or starts at $t+\frac{1}{\varepsilon}$ while we allocate $m$ jobs with processing time $p_j$  to a separate machine each starting at time $t$. Since the right hand side of Eq.~(\ref{eq:comp_ratio}) is increasing with increasing $j$, there is no benefit in considering a processing time $p_j$ that is not larger than processing time $p_{j-1}$. Otherwise we can reduce the processing time $p_{j-1}$ and obtain an larger value $c$.

To minimize $c$, we assume a value $p_2$ with $1<p_2<\frac{1}{\varepsilon}$ and $p_{m+1}=\frac{1}{\varepsilon}$. We justify the second assumption by the fact that the difference between $p_{m+1}$ and $\frac{1}{\varepsilon}$ is arbitrarily small. We transform Eq.~(\ref{eq:comp_ratio}) into a set of equations and then subtract the equation for $p_{j-1}$ from the equation for $p_j$:
\begin{align*}
p_2 & = \frac{c-1}{m} \\ 
p_{j} & = p_{j-1}\cdot \left(\frac{c}{m} + 1 \right)\;\; \mbox{ for } \;\; 2<j \leq m+1 \\ 
p_{m+1} & = \frac{1}{\varepsilon}
\end{align*}
Then we must solve the equation $\frac{c}{m}=\left( \frac{m}{(c-1)\cdot \varepsilon}\right)^{\frac{1}{m-1}}-1$. The condition $p_2=\frac{c-1}{m} $is necessary to obtain the largest value for $c$ and guarantee the validity of Eq.~(\ref{eq:comp_ratio}) for $j=2$. 
\qed
\end{proof}
For large values of $m$, we obtain approximately $\frac{c}{m}=\left(\frac{1}{\varepsilon} \right)^{\frac{1}{m}}$.

For $m=1$, the lower bound of Theorem~\ref{thm:lbound_nonprmp} becomes $1+\frac{1}{\varepsilon}$, that is the lower bound of the corresponding preemptive problem. Theorems~\ref{thm:nonpreemptive_online} and \ref{thm:lbound_nonprmp} show that there is a gap  between the lower bound of the competitive ratio and the competitive ratio of Algorithm~\ref{alg:nonpreemptive}. Therefore, the lower bound of Theorem~\ref{thm:lbound_nonprmp} is not the best possible lower bound for this problem. Deriving a better lower bound that is applicable to all values of $m$ remains a challenging open problem.

\begin{acknowledgement}

This work was partially funded by ERC Advanced Grant 788893 AMDROMA.

\end{acknowledgement}

\bibliographystyle{spmpsci}
\bibliography{schwiegelshohn}

\begin{thebibliography}{10}
\providecommand{\url}[1]{{#1}}
\providecommand{\urlprefix}{URL }
\expandafter\ifx\csname urlstyle\endcsname\relax
  \providecommand{\doi}[1]{DOI~\discretionary{}{}{}#1}\else
  \providecommand{\doi}{DOI~\discretionary{}{}{}\begingroup
  \urlstyle{rm}\Url}\fi

\bibitem{AAB18}
Alon, N., Azar, Y., Berlin, M.: The price of bounded preemption.
\newblock In: Proceedings of the 30th on Symposium on Parallelism in Algorithms
  and Architectures, SPAA '18, pp. 301--310. ACM, New York, NY, USA (2018)

\bibitem{AwerbuchAFLR01}
Awerbuch, B., Azar, Y., Fiat, A., Leonardi, S., Ros{\'{e}}n, A.: On-line
  competitive algorithms for call admission in optical networks.
\newblock Algorithmica \textbf{31}(1), 29--43 (2001)

\bibitem{AKLMNY15}
Azar, Y., Kalp{-}Shaltiel, I., Lucier, B., Menache, I., Naor, J., Yaniv, J.:
  Truthful online scheduling with commitments.
\newblock In: Proceedings of the Sixteenth {ACM} Conference on Economics and
  Computation, {EC} '15, Portland, OR, USA, June 15-19, 2015, pp. 715--732
  (2015)

\bibitem{BaH97}
Baruah, S., Haritsa, J.: Scheduling for overload in real-time systems.
\newblock IEEE Trans. Computers \textbf{46}(9), 1034--1039 (1997)

\bibitem{BKMMRRSW92}
Baruah, S.K., Koren, G., Mao, D., Mishra, B., Raghunathan, A., Rosier, L.E.,
  Shasha, D., Wang, F.: On the competitiveness of on-line real-time task
  scheduling.
\newblock Real-Time Systems \textbf{4}(2), 125--144 (1992)

\bibitem{BuG10}
{Bunde}, D., {Goldwasser}, M.: Dispatching equal-length jobs to parallel
  machines to maximize throughput.
\newblock Lecture Notes in Computer Science \textbf{6139}, 346--358 (2010)

\bibitem{CaI98}
Canetti, R., Irani, S.: Bounding the power of preemption in randomized
  scheduling.
\newblock {SIAM} J. Comput. \textbf{27}(4), 993--1015 (1998)

\bibitem{CEMSS18}
Chen, L., Eberle, F., Megow, N., Schewior, K., Stein, C.: A general framework
  for handling commitment in online throughput maximization.
\newblock CoRR \textbf{abs/1811.08238} (2018).
\newblock \urlprefix\url{http://arxiv.org/abs/1811.08238}

\bibitem{DaP01}
DasGupta, B., Palis, M.: Online real-time preemptive scheduling of jobs with
  deadlines on multiple machines.
\newblock Journal of Scheduling \textbf{4}(6), 297--312 (2001)

\bibitem{GNYZ02}
Garay, J., Naor, J., Yener, B., Zhao, P.: On-line admission control and packet
  scheduling with interleaving.
\newblock In: Proceedings {IEEE} {INFOCOM} 2002, The 21st Annual Joint
  Conference of the {IEEE} Computer and Communications Societies, New York,
  USA, June 23-27, 2002, pp. 94--103 (2002)

\bibitem{GPS00}
Goldman, S., Parwatikar, J., Suri, S.: Online scheduling with hard deadlines.
\newblock Journal of Algorithms \textbf{34}(2), 370 -- 389 (2000)

\bibitem{Gol99}
Goldwasser, M.: Patience is a virtue: The effect of slack on competitiveness
  for admission control.
\newblock In: Proceedings of the Tenth Annual {ACM-SIAM} Symposium on Discrete
  Algorithms, 17-19 January 1999, Baltimore, Maryland., pp. 396--405 (1999)

\bibitem{Gol03}
Goldwasser, M.: Patience is a virtue: the effect of slack on competitiveness
  for admission control.
\newblock Journal of Scheduling \textbf{6}(2), 183--211 (2003)

\bibitem{Graham}
Graham, R., Lawler, E., Lenstra, J., Kan, A.R.: Optimization and approximation
  in deterministic sequencing and scheduling: a survey.
\newblock In: P.~Hammer, E.~Johnson, B.~Korte (eds.) Discrete Optimization II,
  \emph{Annals of Discrete Mathematics}, vol.~5, pp. 287 -- 326. Elsevier
  (1979)

\bibitem{Hor74}
Horn, W.: Some simple scheduling algorithms.
\newblock Naval Research Logistics (NRL) \textbf{21}(1), 177--185 (1974)

\bibitem{KaP01}
Kalyanasundaram, B., Pruhs, K.: Eliminating migration in multi-processor
  scheduling.
\newblock J. Algorithms \textbf{38}(1), 2--24 (2001)

\bibitem{KaP03}
Kalyanasundaram, B., Pruhs, K.: Maximizing job completions online.
\newblock J. Algorithms \textbf{49}(1), 63--85 (2003)

\bibitem{KiC01}
Kim, J., Chwa, K.Y.: On-line deadline scheduling on multiple resources.
\newblock In: Computing and Combinatorics, 7th Annual International Conference,
  {COCOON} 2001, Guilin, China, August 20-23, 2001, Proceedings, pp. 443--452
  (2001)

\bibitem{KoS94}
Koren, G., Shasha, D.: {MOCA:} {A} multiprocessor on-line competitive algorithm
  for real-time system scheduling.
\newblock Theor. Comput. Sci. \textbf{128}(1{\&}2), 75--97 (1994)

\bibitem{Lawler90}
Lawler, E.: A dynamic programming algorithm for preemptive scheduling of a
  single machine to minimize the number of late jobs.
\newblock Ann. Oper. Res. \textbf{26}(1-4), 125--133 (1991)

\bibitem{Lee03}
Lee, J.: Online deadline scheduling: multiple machines and randomization.
\newblock In: {SPAA} 2003: Proceedings of the Fifteenth Annual {ACM} Symposium
  on Parallelism in Algorithms and Architectures, June 7-9, 2003, San Diego,
  California, {USA} (part of {FCRC} 2003), pp. 19--23 (2003)

\bibitem{LiT94}
Lipton, R., Tomkins, A.: Online interval scheduling.
\newblock In: Proceedings of the Fifth Annual {ACM-SIAM} Symposium on Discrete
  Algorithms. 23-25 January 1994, Arlington, Virginia., pp. 302--311 (1994)

\bibitem{LMNY13}
Lucier, B., Menache, I., Naor, J., Yaniv, J.: Efficient online scheduling for
  deadline-sensitive jobs: extended abstract.
\newblock In: 25th {ACM} Symposium on Parallelism in Algorithms and
  Architectures, {SPAA} '13, Montreal, QC, Canada - July 23 - 25, 2013, pp.
  305--314 (2013)

\bibitem{Por04}
Porter, R.: Mechanism design for online real-time scheduling.
\newblock In: Proceedings 5th {ACM} Conference on Electronic Commerce
  (EC-2004), New York, NY, USA, May 17-20, 2004, pp. 61--70 (2004)

\bibitem{PruhsW07}
Pruhs, K., Woeginger, G.: Approximation schemes for a class of subset selection
  problems.
\newblock Theor. Comput. Sci. \textbf{382}(2), 151--156 (2007)

\bibitem{SchwiegelshohnS16}
Schwiegelshohn, C., Schwiegelshohn, U.: The power of migration for online slack
  scheduling.
\newblock In: 24th Annual European Symposium on Algorithms, {ESA} 2016, August
  22-24, 2016, Aarhus, Denmark, pp. 75:1--75:17 (2016)

\bibitem{Woe94}
Woeginger, G.: On-line scheduling of jobs with fixed start and end times.
\newblock Theor. Comput. Sci. \textbf{130}(1), 5--16 (1994)

\end{thebibliography}

\end{document}